\documentclass[a4paper]{article}
\usepackage{amsfonts}
\usepackage{amsthm}
\usepackage{amsmath}
\usepackage{amssymb}

\usepackage{tikz}
\usepackage{pgfpages}

\usepackage{url}
\usepackage{color}

\usetikzlibrary{patterns}

\newtheorem{theorem}{Theorem}
\newtheorem{lemma}[theorem]{Lemma}

\newtheorem{note}[theorem]{Note}

\pgfdeclarepatternformonly{north west lines2}{%
\pgfqpoint{-1pt}{-1pt}}{\pgfqpoint{4pt}{4pt}}{\pgfqpoint{3pt}{3pt}}%
{
  \pgfsetlinewidth{0.4pt}
  \pgfpathmoveto{\pgfqpoint{0pt}{3pt}}
  \pgfpathlineto{\pgfqpoint{3.1pt}{-0.1pt}}
  \pgfusepath{stroke}
}

\begin{document}

\title{The ``Game about Squares'' is NP-hard}

\author{Jens Ma{\ss}berg\\ \small
 Institut f\"ur Optimierung und Operations
Research,
Universit\"at Ulm,\\
\small  {jens.massberg@uni-ulm.de}}

%\date{}

\newcommand{\Square}[2]{
\draw[fill=blue] (#1,#2) -- ++(0,1) -- ++(1,0) -- ++(0,-1) -- cycle;
%\draw[fill=white] (#1+.05,#2+0.05) -- ++(0,0.9) -- ++(0.9,0) -- ++(0,-0.9) -- 
%cycle;
}

\newcommand{\SquareC}[4]{
 \draw[fill=#4,color=#4] (#1,#2) -- ++(0,1) -- ++(1,0) -- ++(0,-1) -- cycle;
 \node at (#1 +0.3,#2+0.2) {#3};
}

\newcommand{\Block}[2]{
\draw[fill=black] (#1,#2) -- ++(0,1) -- ++(1,0) -- ++(0,-1) -- cycle;
}

\newcommand{\Down}[2]{
   \draw[fill=black] (#1+0.2+0.5,#2+0.1+0.5) -- ++ (-0.2,-0.2) -- ++(-0.2,0.2) 
-- cycle;
}

\newcommand{\Left}[2]{
   \draw[fill=black] (#1+0.1+0.5,#2+0.2+0.5) -- ++ (-0.2,-0.2) -- ++(0.2,-0.2) 
-- cycle;
}

\newcommand{\SquareDown}[2]{
  \draw[fill=white,color=white] (#1+0.2+0.5,#2+0.1+0.5) -- ++ (-0.2,-0.2) -- ++(-0.2,0.2) 
    -- cycle;
}

\newcommand{\SquareLeft}[2]{
  \draw[fill=white, color=white] (#1+0.1+0.5,#2+0.2+0.5) -- ++ (-0.2,-0.2) -- ++(0.2,-0.2) 
    -- cycle;
}

\newcommand{\Final}[3]{
   \node[draw, circle, fill=blue] at (#1+0.5,#2+0.5) {};
}
\newcommand{\FinalC}[4]{
   \node[draw, circle, fill=#4, color=#4] at (#1+0.5,#2+0.5) {};
   \node at (#1+0.5,#2+0.5) {#3};
}

\definecolor{lightgrayX}{rgb}{0.9, 0.9, 0.9}
\definecolor{lightgray}{rgb}{0.8, 0.8, 0.8}
%\definecolor{colorA}{HTML}{4BBCF6}
%\definecolor{colorB}{HTML}{98E466}
%\definecolor{colorC}{HTML}{FBEF69}
%\definecolor{colorD}{HTML}{FA6666}

\definecolor{colorA}{HTML}{ED7F22}
\definecolor{colorB}{HTML}{2D3E50}
\definecolor{colorC}{HTML}{297FB8}
\definecolor{colorD}{HTML}{C1392B}

\maketitle


\begin{abstract}
 In the ``Game about Squares'' the task is to push unit squares on an integer
 lattice onto corresponding dots. A square can only be moved into one given
 direction. 
 When a square is pushed onto a lattice point with an arrow the direction of the
 square adopts the direction of the arrow. Moreover, squares can push other
 squares.
 
 In this paper we study the decision problem, whether all squares can be moved
 onto their corresponding dots by a finite number of pushes. We prove that this
 problem is NP-hard.
\end{abstract}

\section{Introduction}

The ``Game about Squares''  \cite{gas} is an addictive game where unit squares 
have to be moved on an integer lattice onto dots of the same color. It has been released by Andrey
 Shevchuk in July 2014.
In the meantime several clones of the game are available for different platforms.

%addictive  clever, hard.

The basic rules of the game are the following:
%The game is played on an uniform grid. Initially, there is a set of squares of different colors that are  
%placed on different grid points. For each square there exists a corresponding grid point (the final 
%position) marked by a circle of the color of the square. The task is to move the square onto their
% corresponding final position.
%Each square has an initial direction (left, right, up or down).
%In each round of the game the player pushes a square. This square moves by one position
%of the grid into the direction $d$ of the square. If this grid point already contains a square, this 
%square is also moved by one position into direction $d$, independent of its own direction.
%... move more ...

Let $\mathcal{D}=\{(-1,0),(1,0),(0,-1),(0,1)\}$ (left, right, down and up, respectively) be a set of directions.
The game is played on an infinite integer lattice $\mathbb{Z}^2$.
An instance $(S,p(S),f(S),d(S),A,p(A),d(A))$ of the game consists of
\begin{itemize}
\item A finite set of squares $S$ with different initial positions $p:S\rightarrow \mathbb{Z}^2$ on
 the lattice. In the game the squares are represented by unit squares of different colors.
\item A final position $f:S\rightarrow \mathbb{Z}^2$ for every square, marked by a
dot of the color of the corresponding square. No two squares have the same final position.
\item An initial direction $d:S\rightarrow \mathcal{D}$ for every square.
\item A finite set of arrows $A$ with distinct positions $p:A\rightarrow \mathbb{Z}^2$ and 
 directions $d:A\rightarrow \mathcal{D}$.
\end{itemize}

The game is played in rounds.
In every round the player chooses a square $s\in S$ that is pushed. Let $d$ be the direction $d(s)$ of 
the square. The square moves one position into direction $d$, that is, its new position is 
$p(s)_{\text{new}}:= p(s)+d$.
If there is already another square $s_2$ at the new position, $s_2$ also moves into direction $d$,
independent of its own direction. If $s_2$ lands on the position of a third square $s_3$, $s_3$ also  
moves into direction $d$ and so on. If a square $s$ lands on a position with an arrow 
(that is, there exists an $a\in A$ with $p(s)_{\text{new}}=p(a)$), the square adopts the direction 
of the arrow, that is, $d_{\text{new}}(s):= d(a)$.

The player wins the game if after a finite number of moves each square $s\in S$ is on its final 
position $f(s)$.
A \emph{winning sequence} is a sequence $(s_1,\ldots, s_k)$, $s_i\in S$, such that if in each round 
$i\in\{1,\ldots,k\}$ the player pushes the square $s_i$,  the game is won in round $k$.

The original game \cite{gas} consists of 35 levels with increasing difficulty.
Between Level 22 and Level 23 the author of the game, Andrey Shevchuk, asks ``Do you think 
this game is hard?''.  We interpret this question in a mathematical way (even if this has not been the
 intention of Shevchuk).
We prove by a reduction from {\sc Satisfiability}, that the game is NP-hard.
Nevertheless, it remains an open question, if this game is in NP or if it is even PSPACE-hard.

\section{Reduction from SATISFIABILITY}

We prove that the game is NP-hard by a reduction from {\sc SATISFIABILITY},
which has been proven
 to be NP-hard by Cook \cite{Cook}. More precisely we prove, that it is NP-hard to decide if a given
 instance of the ``Game about Squares'' (in short GaS) can be won, that is, there is a finite number of 
 moves such that all squares reach their final position.

Let $(X,\mathcal{C})$ be a {\sc Satisfiability} instance where $X=\{x_1,\ldots, x_n\}$ is a set of variables and $\mathcal{C}=\{C_1,\ldots, C_m\}$ is a set of clauses over $X$.
We construct an instance for the GaS that can be won if and only if there is a truth assignment 
$\pi:X\rightarrow\{\mathtt{true},\mathtt{false}\}$ satisfying all clauses.

\begin{figure}[ht]
\begin{center}
 \begin{tikzpicture}[scale=0.8]
   \draw[draw=white,fill=lightgrayX] (1,-4) rectangle (2,2);
   \draw[draw=white,fill=lightgrayX] (3,-4) rectangle (4,1);
 
   \draw[dotted] (0.5,-4.5) grid (5.5,2.5);
   
   \Block{2}{0};

   \node at (1.5,2.2) {$x_i$};
   \node at (3.5,2.2) {$\overline{x_i}$};

   \SquareC{4}{1}{$p$}{colorA};
   \SquareLeft{4}{1};
%   \node at (4.5,1.2) {A};
   \FinalC{2}{1}{$p$}{colorA};
   
   \SquareC{3}{1}{\textcolor{white}{$\mathbf{x}$}}{colorB};
   \SquareDown{3}{1};
   \FinalC{1}{-4}{\textcolor{white}{$\mathbf{x}$}}{colorB};
     
   \Left{3}{-4};

  \end{tikzpicture}
\end{center}
 \caption{Variable gadget for a variable $x_i$.
 A black triangle shows the position of an arrow and its direction. White triangles
 show the initial direction of a square.
   If the left column is 
used by the square labeled with $x$ then $x_i=\mathtt{true}$. Otherwise we have $x_i=\mathtt{false}$.}
  \label{fig:variable}
\end{figure}
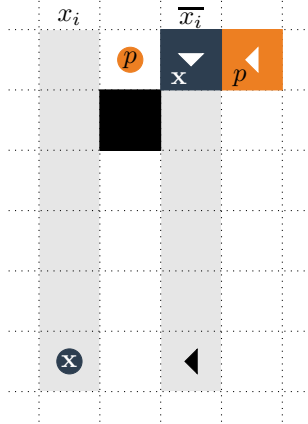

In our GaS instance squares can only be moved to the left and down. To this end, the
initial direction of each square and the direction of each arrow is either left or down.
With this restriction we observe, that a square can never be above or left of its initial position. If a  
square if below or right of its final position, the game cannot be won.
For a given square $s$ we call a position infeasible, if it it left or below of the final
position of $s$ or if it is above or right of the initial position of $s$. Otherwise, we call the position 
feasible for $s$.

Moreover, we use in our instance so called blockers, that are squares that are initially at their
final position.  Moving them from that position, they can never be moved back to their final position.
Thus in order to win the game, blockers are not allowed to be moved.

For every variable we insert a variable gadget as shown in Figure \ref{fig:variable}.
It consists of a variable square (labeled $x$), a decision square (labeled $p$), a blocker and two 
columns. The decision square can push the variable square from the 
right to the left column, where the blocker ensures that these are the only two columns
that can be used.  Depending on which column the variable square moves 
to its final position, the variable is set to $\mathtt{true}$ or to $\mathtt{false}$.
 Accordingly we associate the  literal $x_i$ with the left and $\overline{x_i}$ with right column.

\begin{figure}[ht]
 \begin{center}
 \begin{tikzpicture}[scale=0.8]
   \draw[fill=lightgrayX,draw=white] (3,4) rectangle (10,5);
   \draw[fill=lightgrayX,draw=white] (2,6) rectangle (10,7);

   \draw[dotted] (1.5,2.5) grid (10.5,7.5);

   \SquareC{9}{6}{$C_j$}{colorC};
   \SquareLeft{9}{6};
   \FinalC{2}{4}{\tiny $C_j$}{colorC};

   \FinalC{3}{3}{\tiny $D_j$}{colorD};
   \SquareC{4}{4}{$D_j$}{colorD};
   \SquareDown{4}{4};

   \Down{2}{6};

 \end{tikzpicture}
 \end{center}
 \caption{Clause gadget for a clause $C_j$. The two rows are colored in gray.
 The clause squares has to go to its final position on the 
 left.
 The square $D$ can only reach its final destination if the clause square
 uses the row at the bottom, that is, the clause is satisfied.
 }
 \label{fig:clause}
\end{figure}
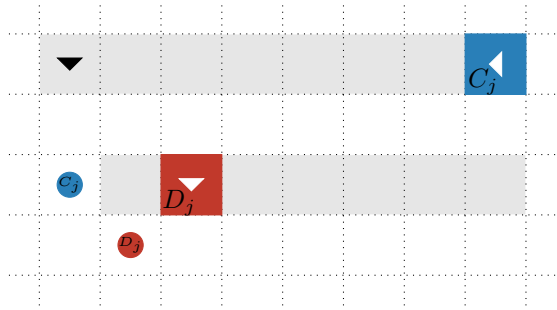

For every clause $C_j$ we insert a clause gadget as shown in Figure \ref{fig:clause}.
It consists of two rows, a clause square $C_j$ and a square $D_j$ indicating if 
the clause is satisfied or not.
Only if the clause square moves on the lower row the indication square $D_j$ can 
reach its final position.

We build a lattice containing these gadgets such that each pair of variable columns 
intersects with each pair of clause rows:
 For $i\in\{1,\ldots,n\}$ we place the variable gadget for $x_i$
in such a way that the variable square is at $(4(i+1),4(m+1))$ and its final position is at
 $(4(i+1)-2,1)$.
 For $j\in\{1,\ldots,m\}$ we place the clause gadget for $C_j$ such that the clause square
 is at  $(4(n+1),4(m-j)+6)$ and its final position is at $(1,4(m-j)+4)$. 
 See Figure \ref{fig:all} for an example.
All gadgets are placed into a lattice of size $4(n+1)\times 4(m+1)$.
Note that each lattice point is feasible for one variable and one clause square.
It remains to specify the crossings of clause rows and variable columns.

\begin{figure}[ht]
\begin{center}
 \begin{tikzpicture}[scale=0.8]
   \draw[fill=lightgrayX,draw=white] (1,-0.5) rectangle (2,5.5);
   \draw[fill=lightgrayX,draw=white] (3,-0.5) rectangle (4,5.5);
   \draw[fill=lightgray,draw=white] (-0.5,1) rectangle (5.5,2);
   \draw[fill=lightgray,draw=white] (-0.5,3) rectangle (5.5,4);
 
   \draw[dotted] (-0.5,-0.5) grid (5.5,5.5);
   
   \node at (1.5,5.2) {$x_i$};
   \node at (3.5,5.2) {$\overline{x_i}$};

   \node at (6.5,3.5) {$C_j$ not sat.};
   \node at (6.5,1.5) {$C_j$ sat.};

   \Left{0}{1};
   \Left{0}{3};
   \Left{2}{1};
   \Left{2}{3};
   \Down{1}{0};
   \Down{1}{2};
   \Down{3}{0};
   \Down{3}{2};
 \end{tikzpicture}
 \end{center}
 \caption{A crossing between a variable $x_i$ and a clause $C_j$ that
 neither contains $x_i$ nor $\overline{x_i}$.}
 \label{crossing}
\end{figure}
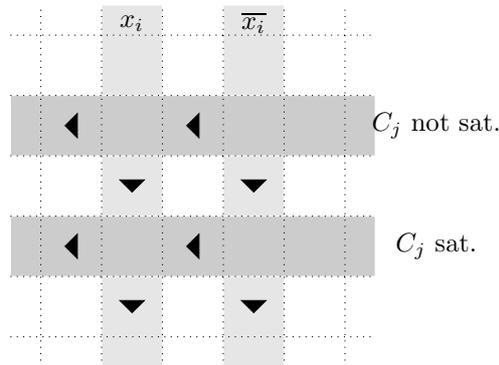

Figure \ref{crossing} shows a crossing between the columns of a variable $x_i$ and a clause 
$C_j$ that neither contains $x_i$ nor $\overline{x_i}$.
 Note that if a variable square pushed a clause tile ore vice versa, the pushed 
square changes its orientation and cannot reach its final position.

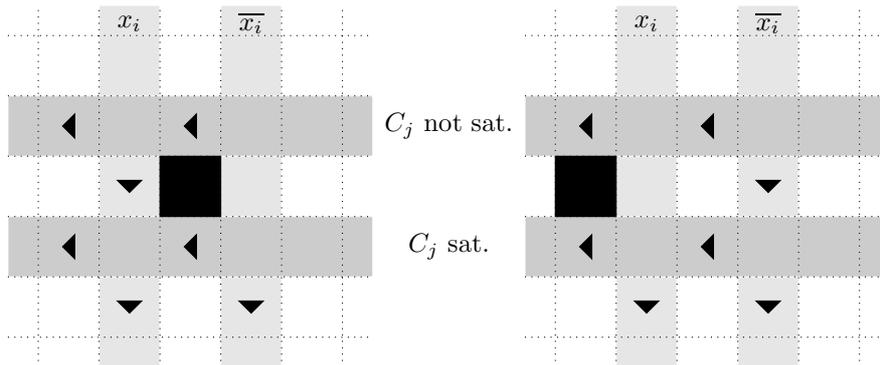
\begin{figure}[ht]
\begin{center}
 \begin{tikzpicture}[scale=0.8]
   \draw[fill=lightgrayX,draw=white] (1,-0.5) rectangle (2,5.5);
   \draw[fill=lightgrayX,draw=white] (3,-0.5) rectangle (4,5.5);
   \draw[fill=lightgray,draw=white] (-0.5,1) rectangle (5.5,2);
   \draw[fill=lightgray,draw=white] (-0.5,3) rectangle (5.5,4);

   \draw[dotted] (-0.5,-0.5) grid (5.5,5.5);

   \node at (1.5,5.2) {$x_i$};
   \node at (3.5,5.2) {$\overline{x_i}$};

   \Block{2}{2};

   \Left{0}{1};
   \Left{0}{3};
   \Left{2}{1};
   \Left{2}{3};
   \Down{1}{0};
   \Down{1}{2};
   \Down{3}{0};
   
   \begin{scope}[xshift=8.5cm, yshift=0cm]
     \draw[fill=lightgrayX,draw=white] (1,-0.5) rectangle (2,5.5);
     \draw[fill=lightgrayX,draw=white] (3,-0.5) rectangle (4,5.5);
     \draw[fill=lightgray,draw=white] (-0.5,1) rectangle (5.5,2);
     \draw[fill=lightgray,draw=white] (-0.5,3) rectangle (5.5,4);

     \node at (1.5,5.2) {$x_i$};
     \node at (3.5,5.2) {$\overline{x_i}$};

     \draw[dotted] (-0.5,-0.5) grid (5.5,5.5);
     \Block{0}{2};

     \Left{0}{1};
     \Left{0}{3};
     \Left{2}{1};
     \Left{2}{3};
     \Down{1}{0};
     \Down{3}{0};
     \Down{3}{2};
   
   \end{scope}

      \node at (6.75,3.5) {$C_j$ not sat.};
   \node at (6.75,1.5) {$C_j$ sat.};

 \end{tikzpicture}
 \end{center}
 \caption{A crossing between a clause $C_j$ that contains $\overline{x_i}$ (left) or 
 $x_i$ (right), respectively.}
 \label{fig:cross2}
\end{figure}

Finally, Figure \ref{fig:cross2} shows the crossing of a variable $x_i$ and a clause $C_j$ that 
contains $x_i$ (on the left) and clause $C_j$ that contains $\overline{x_i}$ (on the right),
respectively. Note that in these cases, the clause square can be pushed to the lower row
by the variable square without changing its direction if and only if the corresponding
literal is satisfied, that is,
the variable square uses the column of the corresponding literal. Once again, the blockers
ensure that the clause squares can leave the crossing only on one of the two clause rows.

We can assume w.l.o.g. that no clause contains both $x_i$ and $\overline{x_i}$ as such clauses are always satisfied.

\begin{figure}[ht]
  \begin{tikzpicture}[scale=0.6]
    
    \foreach \i/\p in {0/4,4/3,8/2,12/1} {
     \draw[fill=lightgray,color=lightgray] (0,5+\i) rectangle (19,6+\i);
     \draw[fill=lightgray,color=lightgray] (1,3+\i) rectangle (19,4+\i);
     \Down{0}{5+\i};
     \SquareC{2}{3+\i}{\tiny $D_\p$}{colorD};
     \SquareDown{2}{3+\i};
     \FinalC{1}{2+\i}{\tiny $D_\p$}{colorD};
    
     \SquareC{19}{5+\i}{\tiny $C_\p$}{colorC};
     \FinalC{0}{3+\i}{\tiny $C_\p$}{colorC};
     \SquareLeft{19}{5+\i};
    }

    %variable
    \foreach \i/\p in {0/1,4/2,8/3,12/4} {
      \draw[fill=lightgray,color=lightgrayX] (4+\i,0) rectangle (5+\i,20);
      \draw[fill=lightgray,color=lightgrayX] (6+\i,0) rectangle (7+\i,19);

      \FinalC{4+\i}{0}{\tiny \textcolor{white}{$\mathbf{x_\p}$}}{colorB};
      \Left{6+\i}{0};
    
      \SquareC{6+\i}{19}{\textcolor{white}{$\mathbf{x_\p}$}}{colorB};
      \SquareDown{6+\i}{19};
     
      \FinalC{5+\i}{19}{\tiny $p_\p$}{colorA};
      \SquareC{7+\i}{19}{$p_\p$}{colorA};
      \SquareLeft{7+\i}{19};
      \Block{5+\i}{18};
    }

%    \draw[dotted] (0,0) grid (20,20);

     \foreach \x in {0,4,8,12,16} {
       \draw (3,2+\x) -- (19,2+\x);
       \draw (3+\x,2) -- (3+\x,18);
     }

    %         % crossing
    \foreach \x/\y in {0/0,4/8, 8/4, 8/12, 12/12} {
        \Left{\x+3}{\y+5};
        \Left{\x+3}{\y+3};
        \Left{\x+5}{\y+5};
        \Left{\x+5}{\y+3};
        
        \Down{\x+4}{\y+4};
        \Down{\x+4}{\y+2};
        \Down{\x+6}{\y+4};
        \Down{\x+6}{\y+2};
    }

    % x_i
    \foreach \x/\y in {0/12,4/12, 0/8, 12/8, 4/0, 12/0 } {
        \Left{\x+3}{\y+5};
        \Left{\x+3}{\y+3};
        \Left{\x+5}{\y+5};
        \Left{\x+5}{\y+3};
        
%        \Down{\x+4}{\y+4};
        \Down{\x+4}{\y+2};
        \Down{\x+6}{\y+4};
        \Down{\x+6}{\y+2};
        \Block{\x+3}{\y+4};
    }

    % not x_i
    \foreach \x/\y in {8/8, 0/4, 4/4, 12/4, 8/0} {
        \Left{\x+3}{\y+5};
        \Left{\x+3}{\y+3};
        \Left{\x+5}{\y+5};
        \Left{\x+5}{\y+3};
        
        \Down{\x+4}{\y+4};
        \Down{\x+4}{\y+2};
%        \Down{\x+6}{\y+4};
        \Down{\x+6}{\y+2};
        \Block{\x+5}{\y+4};
    }

    \end{tikzpicture}
    \label{fig:all}
  \caption{GaS instance for the Satisfiability instance 
  $(x_1\vee x_2)\wedge(x_1\vee\overline{x_3}\vee x_4)\wedge(\overline{x_1}\vee\overline{x_2}\vee\overline{x_4})\wedge(x_2\vee\overline{x_3}\vee x_4)$.}
\end{figure}
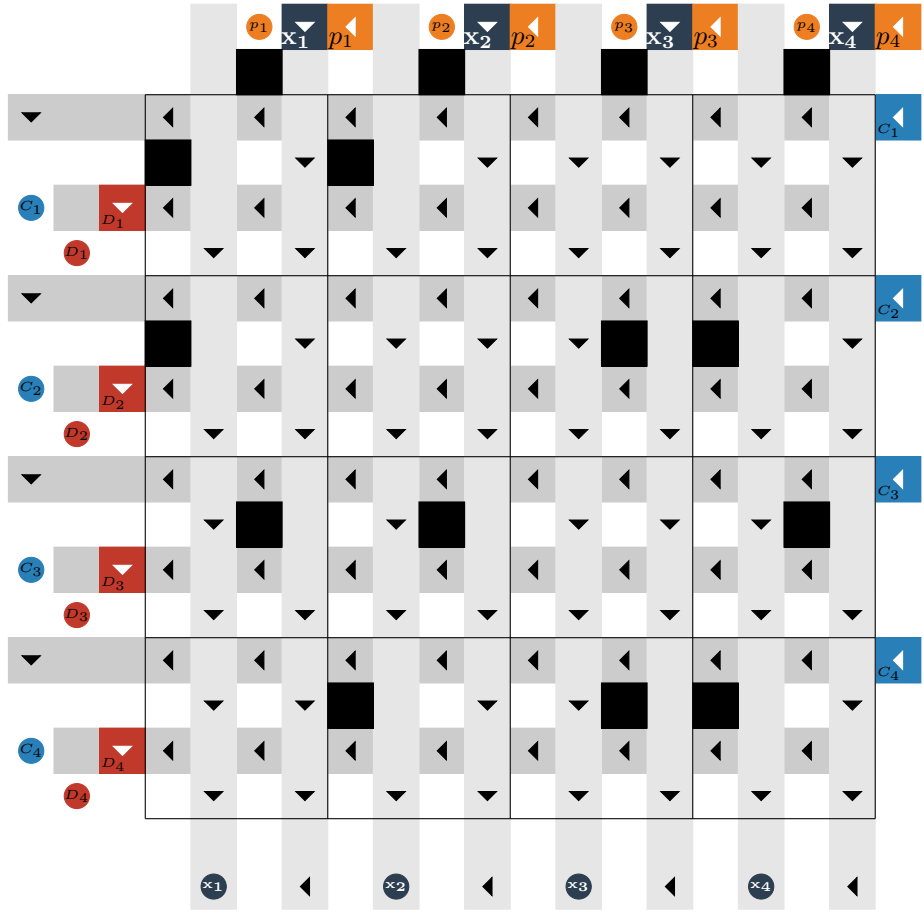

\begin{lemma}\label{lemma1}
  GaS is NP-hard.
\end{lemma}
\begin{proof}
 For a given {\sc SATISFIABILITY} instance $S$ with $n$ variables and $m$ clauses
  we construct a GaS instance as shown above.
 Obviously, this is a polynomial transformation: the GaS instance is placed on a grid of total size $O(nm)$ and contains $O(n+m)$ squares and $O(nm)$  arrows.
 
 We have seen that each variable square has to use either its left or right column 
 in order to win the game. Once such a square has left the uppermost row it cannot be moved
 to the left until it has reached the lowermost row. Otherwise, the square would change its direction
 and cannot reach its final position.
 
 Assume that there exists a truth assignment $\pi$ satisfying $S$. Using the decision squares we push
  each variable square $x$ to the columns corresponding to $\pi(x)$.
 As the truth assignment is satisfied, for each clause $C$ there exists a literal $l\in C$ that is 
 set to $\mathtt{true}$ by $\pi$. Now we push the clause square $C$ to the left
 until we reach the column assigned to $l$. Moving the variable squares down they can be used to 
 push the clause square into their lower rows. Now they can be used to push the squares $D$ by one 
 position to the left so that they can reach their final destination. Finally, all remaining squares can 
 reach their final position without problems and the game is won.

 Now assume, that we have an instance of the game that can be won. We define a truth assignment
 $\pi$ by setting $\pi(x_i)=\mathtt{true}$ if the variable square $x_i$ uses its left column when
  moving down and $\pi(x_i)=\mathtt{false}$ otherwise.
  Now consider the sequence $Q$ of pushes that is used to win the game.
  As for all $i\in\{1,\ldots, m\}$ the square $D_i$ reaches its final position, it must be pushed 
  by the square $C_i$ by one to the left in one of the rounds. But then $C_i$ has been pushed
  to its lower row, which can only be done by a variable square $x_i$ such that the corresponding
  literal $x_i$ or $\overline{x_i}$ is in $C_i$ and is satisfied by $\pi$. Thus each clause is satisfied by $\pi$.
\end{proof}

\begin{note}
  First note that each instance of the game can be restricted to a grid of size
 $O(|S|(|S|+|A|)\times O(|S|(|S|+|A|)$. If there are more than $|S|$ succeeding rows or columns
 that neither contain squares, final positions nor arrows, we can delete all but $|S|$ of them.
 Thus the size of the grid is polynomially bounded in the size of the input.  
 For the problem to be in NP, we have to show that for every instance that can be won there exists
  a certificate verifiable in polynomial time. Such a certificate could be a winning sequence.
   Unfortunately, it is not known if there always exists a winning sequence of polynominal size.

 Nevertheless, the restricted version of the ``Game about Squares'' with instances,
 where only one horizontal and one vertical  direction are allowed, is in NP: Every square can be 
 pushed  at most $O(|S|(|S|+|A|))$ times and thus the game ends after a polynomial number of 
 rounds.
 By the proof of Lemma \ref{lemma1}, this restricted version is NP-hard.
\end{note}

\section*{Acknowledgment}

The author likes to thank Andrey Shevchuk for this addictive game and Jan
Schneider
for valuable discussions.

\nocite{*}
%\addcontentsline{toc}{section}{Bibliography}
%
\bibliography{game_about_squares}{}
\bibliographystyle{plain}

\end{document}